\documentclass{article}
\usepackage[top=1in, bottom=1in, left=1in, right=1in]{geometry} 
\usepackage{setspace} 
\usepackage{flushend}

\usepackage[cmex10]{amsmath} 
\allowdisplaybreaks
\usepackage{cite} 
\usepackage{hyperref} 
\usepackage{color}
\definecolor{darkblue}{rgb}{.15,0,.7}
\hypersetup{
	colorlinks=true,
	linkcolor=darkblue,
	citecolor=darkblue,
	urlcolor=darkblue
}

\usepackage{enumerate} 
\usepackage[pdftex]{graphicx} 
\usepackage{subcaption}
\usepackage{amssymb} 

\usepackage{amsthm} 
\newtheorem{theorem}{Theorem}
\newtheorem{lemma}{Lemma}
\newtheorem*{lemma*}{Lemma}
\newtheorem*{claim*}{Claim}

\theoremstyle{definition}
\newtheorem{definition}{Definition}

\theoremstyle{remark}

\begin{document}
\title{On the Relation Between Identifiability, Differential Privacy and Mutual-Information Privacy}

\author{Weina Wang, Lei Ying and Junshan Zhang
\thanks{%
The authors are with the School of Electrical, Computer and Energy Engineering, Arizona State University, Tempe, AZ 85281, USA (e-mail: \{weina.wang, lei.ying.2, junshan.zhang\}@asu.edu).}
}

\maketitle

\begin{abstract}
This paper investigates the relation between three different notions of privacy: identifiability, differential privacy and mutual-information privacy. Under a unified privacy--distortion framework, where the distortion is defined to be the expected Hamming distance between the input and output databases, we establish some fundamental connections between these three privacy notions. Given a maximum distortion $D$, define $\epsilon_{\mathrm{i}}^*(D)$ to be the smallest (best) identifiability level, and $\epsilon_{\mathrm{d}}^*(D)$ to be the smallest differential privacy level. We characterize $\epsilon_{\mathrm{i}}^*(D)$ and $\epsilon_{\mathrm{d}}^*(D)$, and prove that $ \epsilon_{\mathrm{i}}^*(D)-\epsilon_X\le\epsilon_{\mathrm{d}}^*(D)\le\epsilon_{\mathrm{i}}^*(D)$ for $D$ in some range, where $\epsilon_X$ is a constant depending on the distribution of the original database $X$, and diminishes to zero when the distribution of $X$ is uniform. Furthermore, we show that identifiability and mutual-information privacy are consistent in the sense that given a maximum distortion $D$ in some range, there is a mechanism that optimizes the identifiability level and also achieves the best mutual-information privacy.
\end{abstract}

\section{Introduction}
Privacy has been an increasing concern in the emerging big data era, particularly with the growing use of personal data such as medical records or online activities for big data analysis. Analyzing these data results in new discoveries in science and engineering, but also puts individual's privacy at potential risks. Therefore, privacy-preserving data analysis, where the goal is to preserve the accuracy of data analysis while maintaining individual's privacy, has become one of the main challenges of this big data era. The basic idea of privacy-preserving data analysis is to add randomness in the released information to guarantee that an individual's information cannot be inferred. Intuitively, the higher the randomness is, the better privacy protection individual users get, but the less accurate (useful) the output statistical information is. While randomization seems to be inevitable, for the privacy-preserving data analysis it is of great interest to quantitatively define the notion of privacy. Specifically, we need to understand the amount of randomness needed to protect privacy while preserving usefulness of the data. To this end, we consider three different notions: identifiability, differential privacy and mutual-information privacy, where identifiability is concerned with the posteriors of recovering the original data from the released data, differential privacy is concerned with the additional information leakage of an individual due to the release of the data, and mutual information measures the amount of information about the original database contained in the released data.

While these three different privacy notions are defined from different perspectives, they are fundamentally related. The focus of this paper is to investigate the fundamental connections between these three different privacy notions in the following setting:
\begin{itemize} 
\item We consider the non-interactive database releasing approach for privacy-preserving data analysis, where a synthetic database is released to the public. The synthetic database is a sanitized version of the original database, on which queries and operations can be carried out as if it was the original database. It is then natural to assume that the synthetic database and the original database are in the same universe so the entries have the same interpretation. Therefore we focus on mechanisms that map an input database to an output synthetic database in the same universe. Specifically, we consider a database consisting of $n$ rows, each of which takes values from a finite domain $\mathcal{D}$ of size $m$. In this paper, the database is modeled as a discrete random variable $X$ drawn from $\mathcal{D}^n$ with distribution $p_X$. A mechanism $\mathcal{M}$ is a mapping from an input database $X$ to an output database $Y$, which is also a random variable with alphabet $\mathcal{D}^n$.

\item We define the \emph{distortion} between the output database and the input database to be the expected Hamming distance.~\footnote{Our study of more general distortion measures is underway.} When the input and output are in the same universe, the Hamming distance measures the number of rows two databases differ on, which directly points to the number of rows that need to be modified in order to guarantee a given privacy level.
\end{itemize}

\begin{figure*}
\centering
\includegraphics[scale=0.75]{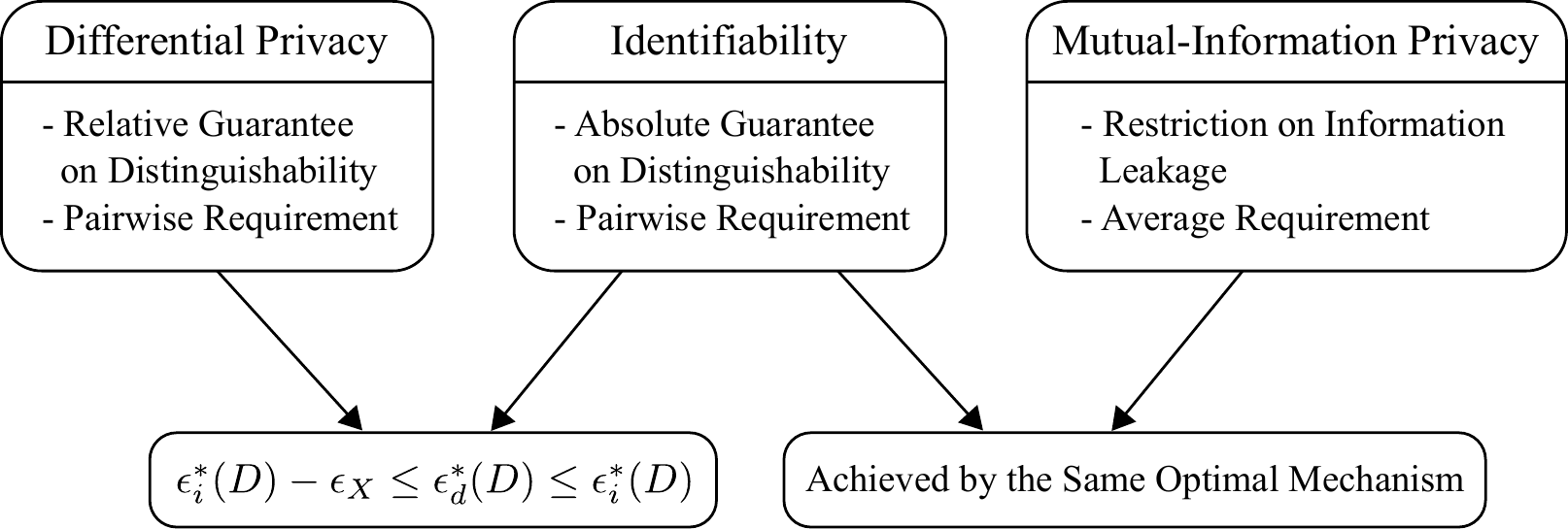}
\caption{Relation between identifiability, differential privacy and mutual-information privacy.}
\label{figRelation}
\end{figure*}
In this paper, we use a unified \emph{privacy--distortion} framework to understand the relation between the three privacy notions. Define the privacy--distortion function to be the best privacy level given a distortion constraint. Then we have the following main results, which are also summarized in \figurename~\ref{figRelation}.

\begin{itemize}
\item[(i)] We derive the exact form of the privacy--distortion function $\epsilon_{\mathrm{i}}^*(D)$ under the notion of identifiability, for some range of the distortion values, by showing that $\epsilon_{\mathrm{i}}^*(D)=h^{-1}(D)$ regardless of the prior distribution, where $$h^{-1}(D)=\ln\bigl(\frac{n}{D}-1\bigr)+\ln(m-1).$$ We further show that for the privacy--distortion function $\epsilon_{\mathrm{d}}^*(D)$ under the notion of differential privacy, $$\epsilon_{\mathrm{i}}^*(D)-\epsilon_X \leq \epsilon_{\mathrm{d}}^*(D)\leq \epsilon_{\mathrm{i}}^*(D),$$ where $\epsilon_X$ is a constant depending on the distribution of $X$ only, given by
\begin{equation*}
\epsilon_X=\max_{x,x'\in\mathcal{D}^n:x\sim x'}\ln\frac{p_X(x)}{p_X(x')}.
\end{equation*}
When the input database has a uniform distribution, we have that $\epsilon_{\mathrm{i}}^*=\epsilon_{\mathrm{d}}^*,$ i.e., differential privacy is equivalent to identifiability. Note that for $\epsilon_X$ to be finite, the distribution $p_X$ needs to have full a support on $\mathcal{D}^n$, i.e., $p_X(x)>0$ for any $x\in\mathcal{D}^n$. When $\epsilon_X$ is large, differential privacy provides only weak guarantee on identifiability. In other words, it is possible to identify some entries of the database with non-trivial accuracy even if the differential privacy is guaranteed when $\epsilon_X$ is large. This is because differential privacy provides a \emph{relative} privacy guarantee, which ensures that limited \emph{additional} information of an individual is leaked in the released data in addition to the knowledge that an adversary has known. Identifiability, on the other hand, guarantees an \emph{absolute} level of indistinguishability of neighboring databases when being inferred from the output database assuming the distribution of $p_X$ and the mechanism are both known to the adversary.

\item[(ii)] Given a maximum distortion $D$ in some range, there is a mechanism that minimizes the identifiability level and also minimizes the mutual information between $X$ and $Y$. In other words, identifiability and mutual-information privacy are consistent under the setting studied in this paper. This is somewhat surprising since identifiability imposes constraints on the distributions of neighboring input databases, which are ``local'' requirements; whereas the mutual information quantifies the correlation strength between the input database and the output database, which is a ``global'' measure. While the two notions are not directly comparable, the fact that they can be optimized simultaneously in the setting studied in this paper reveals the fundamental connection between these two privacy notions.
\end{itemize}

\subsection{Related Work}
Differential privacy, as an analytical foundation for privacy-preserving data analysis, was developed by a line of work (see, e.g., \cite{DwoMcSNis_06,Dwo_06,DwoKenMcS_06}). Dwork et al. \cite{DwoMcSNis_06} proposed the Laplace mechanism which adds Laplace noise to each query result, with noise amplitude proportional to the global sensitivity of the query function. Nissim et al. \cite{NisRasSmi_07} later generalize the mechanism using the concept of local sensitivity. The notion of $(\epsilon,\delta)$-differential privacy \cite{DwoKenMcS_06} has also been proposed as a relaxation of $\epsilon$-differential privacy.

The existing research of differential privacy can be largely classified into two categories: the interactive model where the randomness is added to the result of a query; and the non-interactive model, where the randomness is added to the database before queried. Under the interactive model, a significant body of work has been devoted to privacy--usefulness tradeoff and differentially private mechanisms with accuracy guarantee on each query result have been developed (see, e.g., \cite{GhoRouSun_09,RotRou_10,HarRot_10,MutNik_12}). Since the interactive model allows only a limited number of queries to be answered before the privacy is breached, researchers have also studied the non-interactive model, where synthetic databases or contingency tables with differential privacy guarantees were generated. Mechanisms with distortion guarantee for a set of queries to be answered using the synthetic database have been developed (see, e.g., \cite{BluLigRot_08,DwoNaoRei_09,KasRudSmi_10,UllVad_11,HarLigMcS_12}).

Arising from legal definitions of privacy, identifiability has also been considered as a notion of privacy. Lee and Clifton \cite{LeeCli_12} and Li et al. \cite{LiQarSu_13} proposed differential identifiability and membership privacy, respectively. Mutual information as a measure of privacy leakage has been widely used in the literature (see, e.g., \cite{ClaHunMal_05,Smi_09,ChaChoGuh_10,ZhuBet_05,ChaPalPan_07,duPFaw_12,MakFaw_13,SanRajPoo_13}), mostly under the context of quantitative information flow and anonymity systems.

The connections between different privacy notions have studied recently, e.g., \cite{AlvAndCha_12,duPFaw_12,MakFaw_13,Mir_13}. Alvim et al. \cite{AlvAndCha_12} showed that differential privacy implies a bound on the min-entropy leakage. du Pin Calmon, Makhdoumi and Fawaz \cite{duPFaw_12,MakFaw_13} showed the relationship between differential privacy, information privacy and divergence privacy. Mir \cite{Mir_13} pointed out that the mechanism that achieves the optimal rate--distortion also guarantees a certain level of differential privacy. However, whether this differential privacy level is optimal or how far it is from optimal was not answered in \cite{Mir_13}. The fundamental connection between differential privacy, mutual information and distortion is not yet clear.
The connection between differential privacy and mutual information has also been studied in the two-party setting \cite{McGMirPit_10}, where mutual information is used as the information cost for the protocol of communication between the two parities.

\section{Model}
Consider a database consisting of $n$ rows, each of which corresponds to some sensitive information. For example, each row could be an individual's medical records. The database could also be a graph, where each row indicates the existence of some edge. Suppose that rows take values from a domain $\mathcal{D}$. Then $\mathcal{D}^n$ is the set of all possible values of the database. Two databases, denoted by, $x,x'\in\mathcal{D}^n$, are said to be \emph{neighbors} and denoted as $x\sim x'$ if they differ on exactly row. In this paper, we assume that the domain $\mathcal{D}$ is a finite set and model a database as a discrete random variable $X$ with alphabet $\mathcal{D}^n$ and probability mass function (pmf) $p_X$. Suppose $|\mathcal{D}|=m$ and let $N=m^n$, where $m$ is an integer and $m\ge 2$. Then $|\mathcal{D}^n|=N$. A (randomized) mechanism $\mathcal{M}$ takes a database $x$ as the input, and outputs a random variable $\mathcal{M}(x)$.

\begin{definition}[Mechanism]
A \emph{mechanism} $\mathcal{M}$ is specified by an \emph{associated mapping} $\phi_{\mathcal{M}}\colon \mathcal{D}^n\rightarrow \mathcal{F}$, where $\mathcal{F}$ is the set of multivariate cdf's on some $\mathcal{R}\subseteq\mathbb{R}^r$. Taking database $X$ as the input, the mechanism $\mathcal{M}$ outputs a $\mathcal{R}$-valued random variable $Y$ with $\phi_{\mathcal{M}}(x)$ as the multivariate conditional cdf of $Y$ given $X=x$. \hfill{$\square$}
\end{definition}

In this paper, we focus on mechanisms $\mathcal{M}$ of which the range is the same as the alphabet of $X$, i.e., $\mathcal{R}=\mathcal{D}^n$. Then the output $Y$ is also a discrete random variable with alphabet $\mathcal{D}^n$. Denote the conditional pmf of $Y$ given $X=x$ defined by the cdf $\phi_{\mathcal{M}}(x)$ as $p_{Y\mid X}(\cdot\mid x)$. Throughout this paper we use the following basic notation. We denote the set of real numbers by $\mathbb{R}$, the set of nonnegative real numbers by $\mathbb{R}^+$, and the set of nonnegative integers by $\mathbb{N}$. Let $\overline{\mathbb{R}}^+=\mathbb{R}^+\cup\{+\infty\}$.

\subsection{Different Notions of Privacy}
In addition to the output database $Y,$ we assume that the adversary knows the prior distribution $p_X(x),$ which represents the side information the adversary has, and the privacy-preserving mechanism $\cal M.$ The three notions of privacy studied in this paper are defined next.
\begin{definition}[Identifiability]
A mechanism $\mathcal{M}$ satisfies $\epsilon$\emph{-identifiability} for some $\epsilon\in\overline{\mathbb{R}}^+$ if for any pair of neighboring elements $x,x'\in\mathcal{D}^n$ and any $y\in\mathcal{D}^n$,
\begin{equation}
p_{X\mid Y}(x\mid y)\le e^{\epsilon} p_{X\mid Y}(x'\mid y).
\end{equation} \hfill{$\square$}
\end{definition}
Different from differential identifiability \cite{LeeCli_12} and membership privacy \cite{LiQarSu_13}, which are concerned with whether a particular entity occurs in the database, the notion of identifiability considered here provides a posteriori indistinguishability between any neighboring $x$ and $x'$, thus preventing the whole database from being identified. Note that this is an {\em absolute} level of indistinguishability, such that the adversary cannot distinguish $x$ and $x'$ based on the output database $y$, prior knowledge $p_X$ and the mechanism $\mathcal{M}$.

\begin{definition}[Differential Privacy \cite{DwoMcSNis_06,Dwo_06}]
A mechanism $\mathcal{M}$ satisfies $\epsilon$\emph{-differential privacy} for some $\epsilon\in\overline{\mathbb{R}}^+$ if for any pair of neighboring elements $x,x'\in\mathcal{D}^n$ and any $y\in\mathcal{D}^n$,
\begin{equation}\label{eqDPConditional}
p_{Y\mid X}(y\mid x)\le e^{\epsilon} p_{Y\mid X}(y\mid x').
\end{equation}\hfill{$\square$}
\end{definition}

In \cite{DwoMcSNis_06,Dwo_06},  a mechanism $\mathcal{M}$ satisfies $\epsilon$\emph{-differential privacy} for some $\epsilon\in\overline{\mathbb{R}}^+$ if for any pair of neighboring elements $x,x'\in\mathcal{D}^n$, and any $\mathcal{S}\subseteq\mathcal{R}$,
\begin{equation}\label{eqDPPrior}
\Pr\{Y\in\mathcal{S}\mid X=x\}\le e^\epsilon\Pr\{Y\in \mathcal{S}\mid X=x'\},
\end{equation}
where the conditional probabilities $\Pr\{Y\in\mathcal{S}\mid X=x\}$ and $\Pr\{Y\in\mathcal{S}\mid X=x'\}$ are defined by the multivariate conditional cdf's $\phi_{\mathcal{M}}(x)$ and $\phi_{\mathcal{M}}(x')$, respectively. In the case that the range $\mathcal{R}=\mathcal{D}^n$, which is a discrete set, this is equivalent to the requirement \eqref{eqDPConditional}. The differential privacy property of a mechanism is fully characterized by the associated mapping. Given any particular database, a mechanism $\mathcal{M}$ provides the same privacy guarantee regardless of the prior $p_X$, as long as the associated mapping $\phi_{\mathcal{M}}$ has been specified.

Note that the guarantee that the notion of differential privacy provides is a \emph{relative} one, which ensures that limited \emph{additional} information of an individual is leaked due to the presence of this individual in the database in addition to the knowledge that an adversary has known.

\begin{definition}[Mutual-Information Privacy]
A mechanism $\mathcal{M}$ satisfies $\epsilon$\emph{-mutual-information privacy} for some $\epsilon\in\overline{\mathbb{R}}^+$ if the mutual information between $X$ and $Y$ satisfies $I(X;Y)\leq \epsilon,$ where 
\begin{equation}
I(X;Y)=\sum_{x,y\in\mathcal{D}^n}p_{X,Y}(x,y)\log\frac{p_{X,Y}(x,y)}{p_X(x)p_Y(y)}.
\end{equation}
\hfill{$\square$}
\end{definition}

Mutual information is widely used to quantify information leakage in the literature, mostly under the context of quantitative information flow and anonymity systems. Under our setting, the information leakage we need to quantify is between the input database $X$ and the output database $Y$. Note that the notion of mutual information is an information theoretic notion of privacy, which measures the \emph{average} amount of information about $X$ contained in $Y.$  When $X$ and $Y$ are independent, $I(X;Y)=0.$ The mutual information is maximized and equal to $H(X)$ when $Y=X.$

\subsection{Distortion}
In this paper, we measure the usefulness of a mechanism by the distortion between the database $X$ and the output $Y$, where smaller distortion yields greater usefulness. Consider the Hamming distance $d\colon\mathcal{D}^n\times\mathcal{D}^n\rightarrow \mathbb{N}$. Viewing elements in $\mathcal{D}^n$ as vectors of $n$ rows, the distance $d(x,x')$ between two elements $x,x'\in\mathcal{D}^n$ is the number of rows they differ on. We define the distortion between $X$ and $Y$ to be the expected Hamming distance
\begin{equation}
\mathbb{E}[d(X,Y)]=\sum_{x\in\mathcal{D}^n}\sum_{y\in\mathcal{D}^n}p_X(x)p_{Y\mid X}(y\mid x)d(x,y).
\end{equation}
The Hamming distance also characterizes the neighboring relation on $\mathcal{D}^n$. Two elements $x,x'\in\mathcal{D}^n$ are neighbors if and only if $d(x,x')=1$.

\subsection{Privacy--Distortion Function}
A privacy--distortion pair $(\epsilon,D)$ is said to be \emph{achievable} if there exists a mechanism $\mathcal{M}$ with output $Y$ such that $\mathcal{M}$ satisfies $\epsilon$-privacy level and $\mathbb{E}[d(X,Y)]\le D$. The \emph{privacy--distortion function} $\epsilon^*\colon\mathbb{R}^+\rightarrow\overline{\mathbb{R}}^+$ is defined by
\begin{equation}
\epsilon^*(D)=\inf\{\epsilon\colon (\epsilon,D)\text{ is achievable}\},
\end{equation}
which is the smallest (best) privacy level given the distortion constraint $\mathbb{E}[d(X,Y)]\le D$. We are only interested in the range $[0,n]$ for $D$ since this is the meaningful range for distortion. The privacy--distortion function depends on the prior $p_X$, which reflects the impact of the prior on the privacy--distortion tradeoff. To characterize the privacy--distortion function, we also consider the \emph{distortion--privacy function} $D^*\colon \overline{\mathbb{R}}^+\rightarrow\mathbb{R}^+$ defined by
\begin{equation}
D^*(\epsilon)=\inf\{D\colon (\epsilon,D)\text{ is achievable}\},
\end{equation}
which is the smallest achievable distortion given privacy level $\epsilon$. 

In this paper we consider three different notions of privacy: identifiability, differential privacy and mutual-information privacy, so we denote the privacy--distortion functions under these three notions by $\epsilon_{\mathrm{i}}^*$, $\epsilon_{\mathrm{d}}^*$ and $\epsilon_{\mathrm{m}}^*$, respectively.

\section{Identifiability versus Differential Privacy}\label{secPDT}

In this section, we establish a fundamental connection between identifiability and differential privacy. Given privacy level $\epsilon_{\mathrm{i}}$ and $\epsilon_{\mathrm{d}}$, the minimum distortion level is the solution to the following optimization problems.

\noindent{\bf The Privacy--Distortion Problem under Identifiability (PD-I):}
\begin{align}
\underset{p_{X\mid Y},\mspace{3mu}p_Y}{\text{min}}\mspace{21mu}&\sum_{x\in\mathcal{D}^n}\sum_{y\in\mathcal{D}^n}p_Y(y)p_{X\mid Y}(x\mid y)d(x,y)\nonumber\\ 
\begin{split}\label{conPostPrivacy}
\text{subject to}\mspace{18mu} &\mspace{9mu} p_{X\mid Y}(x\mid y)\le e^{\epsilon_{\mathrm{i}}} p_{X\mid Y}(x'\mid y),\\
&\mspace{57mu}\forall x,x'\in\mathcal{D}^n\colon x\sim x',y\in\mathcal{D}^n,
\end{split}\\
&\sum_{x\in \mathcal{D}^n} p_{X\mid Y}(x\mid y)=1,\mspace{36mu}\forall y\in\mathcal{D}^n,\\
&\mspace{9mu} p_{X\mid Y}(x\mid y)\geq 0,\mspace{54mu}\forall x,y\in\mathcal{D}^n,\\
\begin{split}\label{conPY}
&\sum_{y\in \mathcal{D}^n} p_{X\mid Y}(x\mid y)p_Y(y)=p_X(x),\\
&\mspace{213mu}\forall x\in\mathcal{D}^n,
\end{split}\\
&\mspace{9mu}p_Y(y)\geq 0,\mspace{120mu}\forall y\in\mathcal{D}^n.
\end{align}

\noindent{\bf The Privacy--Distortion Problem under Differential Privacy (PD-DP):}
\begin{align}
\underset{p_{Y\mid X}}{\text{min}}
\mspace{36mu}&\sum_{x\in\mathcal{D}^n}\sum_{y\in\mathcal{D}^n}p_X(x)p_{Y\mid X}(y\mid x)d(x,y)\nonumber\\ 
\begin{split}\label{conPrivacy}
\text{subject to}\mspace{18mu} &\mspace{9mu} p_{Y\mid X}(y\mid x)\le e^{\epsilon_{\mathrm{d}}} p_{Y\mid X}(y\mid x'),\\
&\mspace{57mu} \forall x,x'\in\mathcal{D}^n\colon x\sim x',y\in\mathcal{D}^n,
\end{split}\\
&\sum_{y\in \mathcal{D}^n} p_{Y\mid X}(y\mid x)=1,\mspace{36mu} \forall x\in\mathcal{D}^n,\\
&\mspace{9mu} p_{Y\mid X}(y\mid x)\geq 0,\mspace{54mu}\forall x,y\in\mathcal{D}^n.
\end{align}

For convenience, we first define two constants $\epsilon_X$ and $\widetilde{\epsilon}_X$ that only depend on the prior $p_X$. Let
\begin{equation}\label{eqEpsX}
\epsilon_X=\max_{x,x'\in\mathcal{D}^n:x\sim x'}\ln\frac{p_X(x)}{p_X(x')},
\end{equation}
which is the maximum prior probability difference between two neighboring databases. For $\epsilon_X$ to be finite, the distribution $p_X$ needs to have full support on $\mathcal{D}^n$, i.e., $p_X(x)>0$ for any $x\in\mathcal{D}^n$. To define $\widetilde{\epsilon}_X$, note that the prior $p_X$ puts some constraints on the posterior probabilities. We say $\{p_{X\mid Y}(x\mid y),x,y\in\mathcal{D}^n\}$ is \emph{feasible} if there exists a pmf $p_Y$ such that it is the marginal pmf of $Y$. Let $\widetilde{\epsilon}_X$ be the smallest $\epsilon$ such that the following posterior probabilities are feasible:
\begin{equation*}
p_{X\mid Y}(x\mid y)=\frac{e^{-\epsilon d(x,y)}}{\bigl(1+(m-1)e^{-\epsilon}\bigr)^n},\quad x,y\in\mathcal{D}^n.
\end{equation*}
For any $p_X$, $\widetilde{\epsilon}_X$ is finite since when $\epsilon\rightarrow +\infty$, the pmf $p_Y=p_X$ is the marginal pmf of $Y$.
Finally we consider the function
\begin{equation}
h^{-1}(D)=\ln\bigl(\frac{n}{D}-1\bigr)+\ln(m-1).
\end{equation}

\begin{figure}
\centering
\includegraphics[scale=0.45]{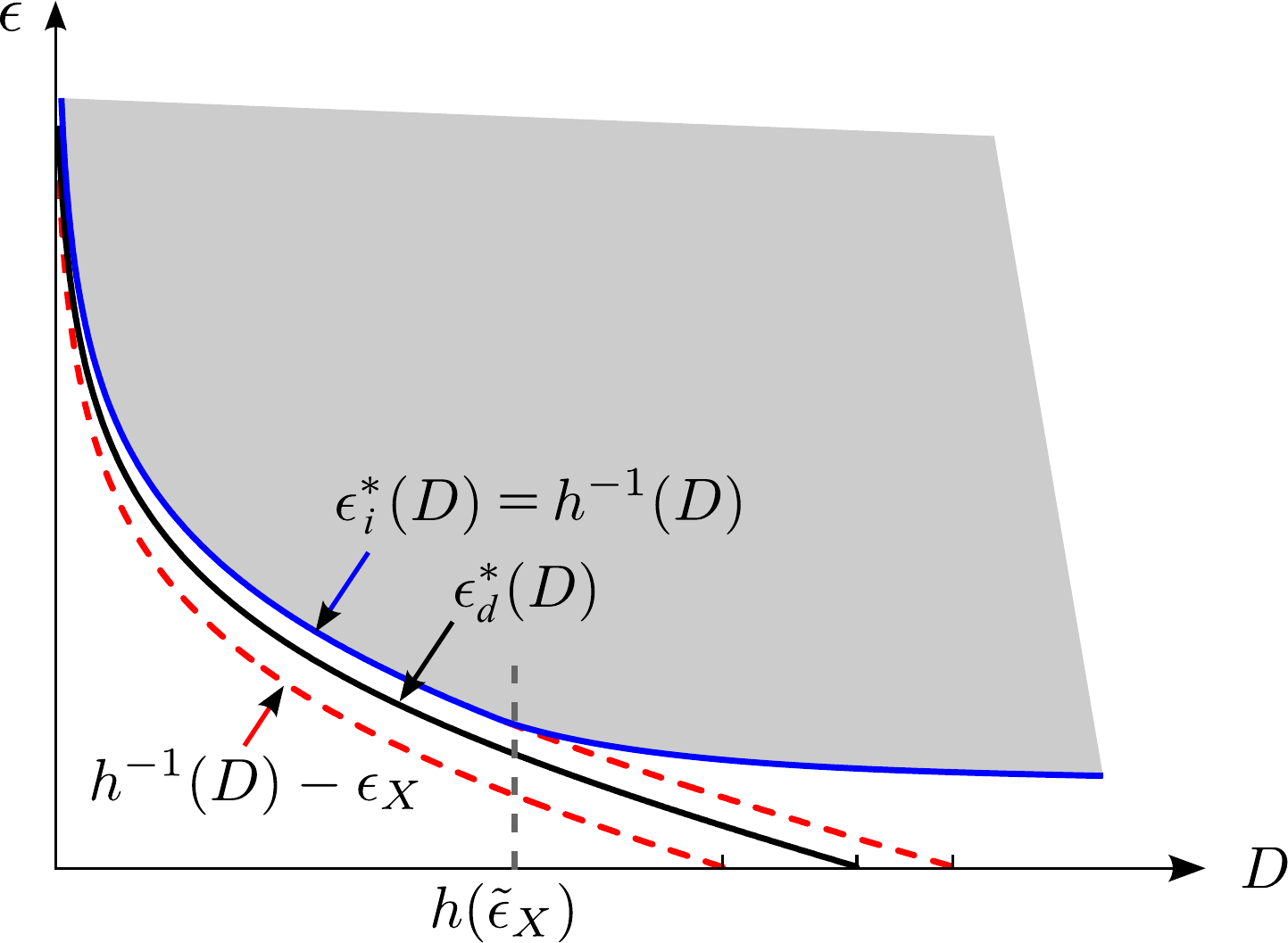}
\caption{The privacy--distortion function $\epsilon_{\mathrm{i}}^*$ under identifiability and $\epsilon_{\mathrm{d}}^*$ under differential privacy satisfy $\epsilon_{\mathrm{i}}^*(D)-\epsilon_X \leq \epsilon_{\mathrm{d}}^*(D)\leq \epsilon_{\mathrm{i}}^*(D)$ for $D$ in some range.}
\label{figBounds}
\end{figure}
Recall that $\epsilon_{\mathrm{i}}^*(D)$ denotes the best identifiability level under a maximum distortion $D$, and $\epsilon_{\mathrm{d}}^*(D)$ denotes the best differential privacy level under a maximum distortion $D$. The connection between the privacy--distortion functions $\epsilon_{\mathrm{i}}^*$ and $\epsilon_{\mathrm{d}}^*$ is established in the following theorem. See \figurename~\ref{figBounds} for an illustration.
\begin{theorem}\label{thm:IDvsDP}
For identifiability, the privacy--distortion function $\epsilon_{\mathrm{i}}^*$ of a database $X$ with $\epsilon_X<+\infty$ satisfies
\begin{equation}
\begin{cases}
\epsilon_{\mathrm{i}}^*(D) =h^{-1}(D),& 0\le D\le h(\widetilde{\epsilon}_X),\\
\epsilon_{\mathrm{i}}^*(D) \ge \max\{h^{-1}(D),\epsilon_X\},& h(\widetilde{\epsilon}_X)<D\le n.
\end{cases}\label{epsiloni}
\end{equation}
For differential privacy, the privacy--distortion function $\epsilon_{\mathrm{d}}^*$ of a database $X$ satisfies the following bounds for any $D$ with $0\le D\le n$:
\begin{equation}\label{epsilond}
\max\{h^{-1}(D)-\epsilon_X,0\}\le\epsilon_{\mathrm{d}}^*(D)\le \max\{h^{-1}(D),0\}.
\end{equation}
\hfill{$\square$}
\end{theorem}

From the theorem above, we can see that $0\leq \epsilon_{\mathrm{i}}^*(D)-\epsilon_{\mathrm{d}}^*(D)\leq \epsilon_X$ when $0\le D\le h(\widetilde{\epsilon}_X)$. The lemmas needed in the proof of this theorem can be found in the appendix. Here we give a sketch of the proof, which consists of the following key steps:
\begin{itemize}
\item The first key step is to show that both PD-I and PD-DP, through (respective) relaxations as shown in \figurename~\ref{figOpt}, boil down to the same optimization problem.

\noindent{\bf Relaxed Privacy--Distortion (R-PD):}
\begin{align}
\underset{p_{X\mid Y},\mspace{3mu}p_Y}{\text{min}}\mspace{21mu}&\sum_{x\in\mathcal{D}^n}\sum_{y\in\mathcal{D}^n}p_Y(y)p_{X\mid Y}(x\mid y)d(x,y)\nonumber\\ 
\begin{split}\label{conRelaxedPrivacy}
\text{subject to}\mspace{18mu} &\mspace{9mu} p_{X\mid Y}(x\mid y)\le e^{\epsilon} p_{X\mid Y}(x'\mid y),\\
&\mspace{57mu}\forall x,x'\in\mathcal{D}^n\colon x\sim x',y\in\mathcal{D}^n,
\end{split}\\
&\sum_{x\in \mathcal{D}^n} p_{X\mid Y}(x\mid y)=1,\mspace{36mu} \forall y\in\mathcal{D}^n,\label{conRelaxedNorm}\\
&\mspace{9mu} p_{X\mid Y}(x\mid y)\geq 0,\mspace{54mu} \forall x,y\in\mathcal{D}^n,\label{conRelaxedNonneg}\\
& \sum_{y\in\mathcal{D}^n}p_Y(y) = 1,\label{conRelaxedNormPY}\\
&\mspace{9mu} p_Y(y)\geq 0, \mspace{120mu} \forall y\in\mathcal{D}^n.\label{conRelaxedNonnegPY}
\end{align}

\begin{figure}
\centering
\includegraphics[scale=0.7]{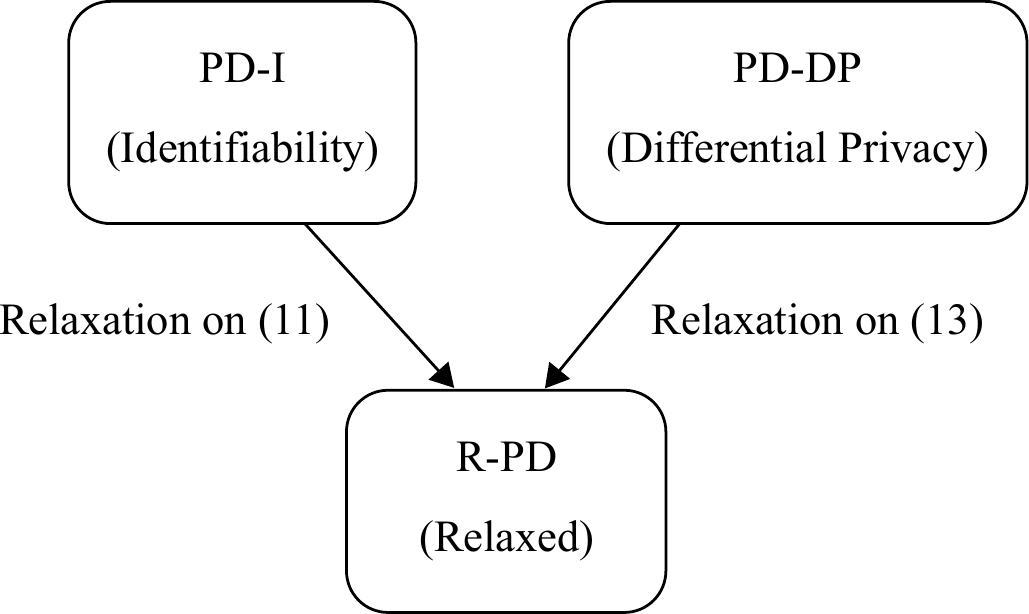}
\caption{Both PD-I and PD-DP boil down to R-PD through different relaxations.}
\label{figOpt}
\end{figure}
Relaxing the constraint \eqref{conPY} in PD-I to the constraint \eqref{conRelaxedNormPY} gives R-PD. Now consider PD-DP. For any neighboring $x,x'\in\mathcal{D}^n$, $p_X(x)\le e^{\epsilon_X}p_X(x')$ according to the definition of $\epsilon_X$, and a necessary condition for the constraint \eqref{conPrivacy} to be satisfied is
\begin{equation}\label{conPrivacyRel}
p_X(x)p_{Y\mid X}(y\mid x)\le e^{\epsilon_{\mathrm{d}}+\epsilon_X} p_X(x')p_{Y\mid X}(y\mid x').
\end{equation}
Therefore, replacing constraint \eqref{conPrivacy} with \eqref{conPrivacyRel} and letting $\epsilon=\epsilon_{\mathrm{d}}+\epsilon_X,$ we obtain R-PD. So R-PD can be regarded as a relaxation of both PD-I and PD-DP.

\item In Lemma~\ref{lem:R-PD}, we prove that the minimum distortion of R-PD is $D_{\mathrm{relaxed}}^*(\epsilon)=h(\epsilon)$, which gives lower bounds on the distortion--privacy functions under identifiability and under differential privacy. By the connection between distortion--privacy function and privacy--distortion function, Lemma~\ref{lem:lowerBounds} shows that $\epsilon_\mathrm{i}^*(D)\ge h^{-1}(D)$ and $\epsilon_{\mathrm{d}}^*(D)\ge h^{-1}(D)-\epsilon_X$ for any $D$ with $0\le D\le n$. Lemma~\ref{lem:lower2ID} shows another lower bound on $\epsilon_{\mathrm{i}}^*$, combining which with the lower bound in Lemma~\ref{lem:lowerBounds} gives the lower bound in Theorem~\ref{thm:IDvsDP}.

\item Consider the mechanism $\mathcal{E}_i$ specified by
\begin{equation}\label{eqPostExpMec}
p_{Y\mid X}(y\mid x)=\frac{p_Y(y)e^{-\epsilon d(x,y)}}{p_X(x)\bigl(1+(m-1)e^{-\epsilon}\bigr)^n},\quad x,y\in\mathcal{D}^n,
\end{equation}
where $\epsilon\ge\widetilde{\epsilon}_X$ and $p_Y$ is the corresponding pmf of $Y$. Lemma \ref{lem:DP-I} shows that the mechanism $\mathcal{E}_i$ guarantees an identifiability level of $\epsilon$ with distortion $h(\epsilon)$ when $\epsilon\ge\widetilde{\epsilon}_X,$  which yields \eqref{epsiloni} when combining with the lower bound above.

\item Consider the mechanism $\mathcal{E}_d$ specified by the conditional probabilities
\begin{equation}\label{eqExpMec}
p_{Y\mid X}(y\mid x)= \frac{e^{-\epsilon d(x,y)}}{\bigl(1+(m-1)e^{-\epsilon}\bigr)^n},\quad x,y\in\mathcal{D}^n,
\end{equation}
where $\epsilon\ge 0$. This is the exponential mechanism with score function $q=-d$ \cite{McSTal_07}. Lemma \ref{lemMecE} shows that the  mechanism $\mathcal{E}$ satisfies $\epsilon$-differential privacy with distortion $h(\epsilon),$ which provides the upper bound in \eqref{epsilond}.
\end{itemize}

\section{Identifiability versus Mutual-Information Privacy}
In this section, we discuss the connection between identifiability and mutual-information privacy. Intuitively, mutual information can be used to quantify the information about $X$ by observing a correlated random variable $Y$. Recall that the privacy--distortion function under mutual-information privacy denotes the smallest achievable mutual information without exceeding a maximum distortion. Note that this formulation has the same form as the formulation in the rate--distortion theory \cite{CovTho_06}, and thus the privacy--distortion function under mutual-information privacy is identical to the rate--distortion function in this setting. We will show that identifiability and mutual-information privacy are consistent under the privacy--distortion framework in the sense that given a maximum distortion $D$, there is a mechanism that minimizes the identifiability level $\epsilon_{\mathrm{i}}^*(D)$ and also achieves the minimum mutual information, for some range of $D$.

It has been pointed out in \cite{Mir_13} that the mechanism that achieves the optimal rate--distortion also guarantees a certain level of differential privacy. However, whether this differential privacy level is optimal or how far it is from optimal was not answered. Our result on the connection between identifiability and mutual-information privacy indicates that given a maximum distortion $D$, there is a mechanism that achieves the optimal rate--distortion and guarantees a differential privacy level $\epsilon$ such that $\epsilon_{\mathrm{d}}^*(D)\le\epsilon\le\epsilon_{\mathrm{d}}^*(D)+\epsilon_X$.

Given a maximum distortion $D$, the privacy--distortion function $\epsilon_{\mathrm{m}}^*(D)$ for input $X$ with pmf $p_X(\cdot)$ is given by the optimal value of the following convex optimization problem.

\noindent\textbf{The Privacy and Distortion Problem under Mutual-Information Privacy (PD-MIP):}
\begin{align}
\underset{p_{Y\mid X}}{\text{min}}
\mspace{36mu}&\mspace{9mu}I(X;Y)\nonumber\\ 
\text{subject to}\mspace{18mu} &\sum_{x\in\mathcal{D}^n}\sum_{y\in\mathcal{D}^n}p_X(x)p_{Y\mid X}(y\mid x)d(x,y)\le D,\label{conDistortion}\\
&\sum_{y\in\mathcal{D}^n}p_{Y\mid X}(y\mid x)=1,\mspace{36mu}\forall x\in\mathcal{D}^n,\label{conNorm}\\
&\mspace{9mu}p_{Y\mid X}(y\mid x)\ge 0,\mspace{54mu}\forall x,y\in\mathcal{D}^n.\label{conNonneg}
\end{align}

\begin{theorem}
For any $D$ with $0\le D\le h(\widetilde{\epsilon}_X)$, the identifiability optimal mechanism $\mathcal{E}_i$ defined in \eqref{eqPostExpMec} is also mutual-information optimal. \hfill{$\square$}
\end{theorem}
\begin{proof}
Consider the Lagrangian and the KKT conditions of the optimization problem PD-MIP. For any $0\le D\le h(\widetilde{\epsilon}_X)$, consider the conditional probabilities $\{p_{Y\mid X}(y\mid x),x,y\in\mathcal{D}^n\}$ in \eqref{eqPostExpMec} under $\mathcal{E}_i$ with $\epsilon=h^{-1}(D)$. Then it is easy to verify that $\{p_{Y\mid X}(y\mid x),x,y\in\mathcal{D}^n\}$ is primal feasible. Let $\lambda=h^{-1}(D)$, $\mu(x)=p_X(x)\ln\bigl[ p_X(x)\bigl(1+(m-1)e^{-\epsilon}\bigr)^n\bigr]$ with $ x\in\mathcal{D}^n$ and $\eta(x,y)=0$ with $x,y\in\mathcal{D}^n$ be the Lagrange multipliers for \eqref{conDistortion}, \eqref{conNorm} and \eqref{conNonneg}, respectively. Then these multipliers are dual feasible. The stationarity condition $p_X(x)\ln{p_{Y\mid X}(y\mid x)}-p_X(x)\ln{p_Y(y)}+\lambda p_X(x)d(x,y)+\mu(x)-\eta(x,y)=0$ (derived in Appendix~\ref{secProofs}) and the complementary slackness condition are also satisfied. Therefore, the above $(p_{Y\mid X},\lambda,\mu,\eta)$ satisfies the KKT conditions of PD-MIP.

Slater's condition holds for the problem PD-MIP since all the inequality constraints are affine \cite{BoyVan_04}. By convexity, the KKT conditions are sufficient for optimality. Therefore the mechanism $\mathcal{E}_i$ with $\epsilon=h^{-1}(D)$ also gives the smallest mutual information.
\end{proof}

\section{Conclusions}
In this paper, we investigated the relation between three different notions of privacy: identifiability, differential privacy and mutual-information privacy, where identifiability provides absolute indistinguishability, differential privacy guarantees limited additional information leakage, and mutual information is an information theoretic notion of privacy. Under a unified privacy--distortion framework, where the distortion is defined to be the Hamming distance between the input and output databases, we established some fundamental connections between these three privacy notions. Given a maximum distortion $D$ within some range, the smallest identifiability level $\epsilon_{\mathrm{i}}^*(D)$ and the smallest differential privacy level $\epsilon_{\mathrm{d}}^*(D)$ are proved to satisfy $\epsilon_{\mathrm{i}}^*(D)-\epsilon_X \leq \epsilon_{\mathrm{d}}^*(D)\leq \epsilon_{\mathrm{i}}^*(D)$, where $\epsilon_X$ is a constant depending on the distribution of the original database, and is equal to zero when the distribution is uniform. Furthermore, we showed that identifiability and mutual-information privacy are consistent in the sense that given the same maximum distortion $D$ within some range, there is a mechanism that simultaneously optimizes the identifiability level and the mutual-information privacy.

Our findings in this study reveal some fundamental connections between the three notions of privacy. With these three notions of privacy being defined, many interesting issues deserve further attention. For example, in some cases, the prior $p_X$ is imperfect, and it is natural to ask how we can protect privacy with robustness over the prior distribution. Some other interesting directions include the generalization from ``pairwise'' privacy to ``group'' privacy, which arises from the pairwise requirements that both identifiability and differential privacy impose on neighboring databases. The connections between membership privacy and these three notions of privacy also need to be explored, since membership privacy has been proposed as a unifying framework for privacy definitions.

\appendix
\section{Proofs}\label{secProofs}
\begin{lemma}\label{lem:R-PD}
The minimum distortion $D_\mathrm{relaxed}^*(\epsilon)$ of the relaxed optimization problem R-PD satisfies
\begin{equation}\label{eqLowerBound}
D_{\mathrm{relaxed}}^*(\epsilon) = h(\epsilon),
\end{equation}
where
\begin{equation}
h(\epsilon) = \frac{n}{1+\frac{e^{\epsilon}}{m-1}}.
\end{equation}
\end{lemma}
\begin{proof}
We first prove the following claim, which gives a lower bound on the minimum distortion $D_\mathrm{relaxed}^*(\epsilon)$.

\begin{claim*}\label{claimLowerBoundFeasible}
Any feasible solution $\{p_{X\mid Y}(x\mid y),x,y\in\mathcal{D}^n\}$ of R-PD satisfies
\begin{equation}
\sum_{x\in\mathcal{D}^n}p_{X\mid Y}(x\mid y)d(x,y)\ge h(\epsilon).
\end{equation}
\end{claim*}

\noindent\textbf{Proof of the Claim.} Consider any feasible $\{p_{X\mid Y}(x\mid y),x,y\in\mathcal{D}^n\}$. For any $y\in\mathcal{D}^n$ and any integer $l$ with $0\le l\le n$, let $\mathcal{N}_l(y)$ be the set of elements with distance $l$ to $y$, i.e.,
\begin{equation}\label{eqNl}
\mathcal{N}_l(y)=\{v\in\mathcal{D}^n\colon d(v,y)=l\}.
\end{equation}
Denote $P_l=\Pr\{X\in\mathcal{N}_l(y)\mid Y=y\}$. Then
\begin{equation*}
\sum_{x\in\mathcal{D}^n}p_{X\mid Y}(x\mid y)d(x,y)=\sum_{l=0}^{n}lP_l.
\end{equation*}

We first derive a lower bound on $P_n$. For any $u\in\mathcal{N}_{l-1}(y)$, $\mathcal{N}_1(u)\cap\mathcal{N}_l(y)$ consists of the neighbors of $u$ that are in $\mathcal{N}_l(y)$. By the constraint \eqref{conRelaxedPrivacy}, for any $v\in\mathcal{N}_1(u)\cap\mathcal{N}_l(y)$,
\begin{equation}\label{ineqPostRelaxed}
p_{X\mid Y}(u\mid y)\le e^{\epsilon}p_{X\mid Y}(v\mid y).
\end{equation}
Each $u\in\mathcal{N}_{l-1}(y)$ has $n-(l-1)$ rows that are the same with the corresponding rows of $y$. Each neighbor of $u$ in $\mathcal{N}_l(y)$ can be obtained by changing one of these $n-(l-1)$ rows to a different element in $\mathcal{D}$, which is left with $m-1$ choices. Therefore, each $u\in\mathcal{N}_{l-1}(y)$ has $(n-l+1)(m-1)$ neighbors in $\mathcal{N}_l(y)$. By similar arguments, each $v\in\mathcal{N}_l(y)$ has $l$ neighbors in $\mathcal{N}_{l-1}(y)$. Taking summation of \eqref{ineqPostRelaxed} over $u\in\mathcal{N}_{l-1}(y),v\in\mathcal{N}_l(y)$ with $u\sim v$ yields
\begin{align*}
&\sum_{u\in\mathcal{N}_{l-1}(y)}(n-l+1)(m-1)p_{X\mid Y}(u\mid y)\nonumber\\
&\le e^{\epsilon}\sum_{u\in\mathcal{N}_{l-1}(y)}\mspace{3mu}\sum_{v\in\mathcal{N}_1(u)\cap\mathcal{N}_l(y)}p_{X\mid Y}(v\mid y).
\end{align*}
Thus
\begin{align}
&\mspace{23mu}(n-l+1)(m-1)P_{l-1}\nonumber\\
&\le e^{\epsilon}\sum_{v\in\mathcal{N}_l(y)}\mspace{3mu}\sum_{u\in\mathcal{N}_1(v)\cap\mathcal{N}_{l-1}(y)}p_{X\mid Y}(v\mid y)\\
&=e^{\epsilon}lP_l.\label{ineqPNeigh}
\end{align}
Recall that $N_l\triangleq|\mathcal{N}_l(x)|=\binom{n}{l}(m-1)^l$. Then by \eqref{ineqPNeigh} we obtain that, for any $l$ with $1\le l\le n$,
\begin{equation*}
\frac{P_{l-1}}{N_{l-1}}\le \frac{P_l}{N_l}e^{\epsilon}.
\end{equation*}
As a consequence, for any $l$ with $0\le l\le n$,
\begin{equation}\label{ineqPIne}
P_l\le \frac{N_l}{N_n}e^{(n-l)\epsilon}P_n.
\end{equation}
Since $\sum_{l=0}^{n}P_l=1$, taking summation over $l$ in \eqref{ineqPIne} yields
\begin{align*}
1&\le P_n\frac{1}{N_n e^{-n\epsilon}}\sum_{l=0}^{n}N_l e^{-l\epsilon}\\
&=P_n\frac{\bigl(1+(m-1)e^{-\epsilon}\bigr)^n}{N_ne^{-n\epsilon}},
\end{align*}
i.e.,
\begin{equation*}
P_n\ge\frac{N_ne^{-n\epsilon}}{\bigl(1+(m-1)e^{-\epsilon}\bigr)^n}.
\end{equation*}
This lower bound on $P_n$ gives the following lower bound:
\begin{align*}
\sum_{l=0}^{n}lP_l&\ge \sum_{l=0}^n l\biggl(P_l+a\frac{N_le^{-l\epsilon}}{\sum_{k=0}^{n-1}N_ke^{-k\epsilon}}\biggr)\nonumber\\
&\mspace{23mu}+\frac{nN_ne^{-n\epsilon}}{\bigl(1+(m-1)e^{-\epsilon}\bigr)^n},
\end{align*}
where $a=P_n-\frac{N_ne^{-n\epsilon}}{\bigl(1+(m-1)e^{-\epsilon}\bigr)^n}$.

Consider the following optimization problem:
\begin{equation*}
\begin{aligned}
\text{min}\mspace{18mu} &&& \sum_{l=0}^{n-1}lQ_l\\
\text{subject to} &&& Q_l \ge 0,\mspace{73mu} l=0,1,\dots,n-1,\\
&&& \frac{Q_{l-1}}{N_{l-1}}\le\frac{Q_l}{N_l}e^{\epsilon},\mspace{18mu} l=1,2,\dots,n-1,\\
&&& \sum_{l=0}^{n-1}Q_l=1-\frac{N_ne^{-n\epsilon}}{\bigl(1+(m-1)e^{-\epsilon}\bigr)^n}.
\end{aligned}
\end{equation*}
Suppose the optimal solution of this problem is $\{Q_0^*,Q_1^*,\dots,Q_{n-1}^*\}$. Then
\begin{equation*}
\sum_{l=0}^{n-1} l\biggl(P_l+a\frac{N_le^{-l\epsilon}}{\sum_{k=0}^{n-1}N_ke^{-k\epsilon}}\biggr)\ge\sum_{l=0}^{n-1}lQ_l^*
\end{equation*}
as $\Bigl\{P_l+a\frac{N_le^{-l\epsilon}}{\sum_{k=0}^{n-1}N_ke^{-k\epsilon}},l=0,1,\dots,n-1\Bigr\}$ is a feasible solution. Therefore,
\begin{equation*}
\sum_{l=0}^{n}lP_l\ge\sum_{l=0}^{n-1}lQ^*_l+\frac{nN_ne^{-n\epsilon}}{\bigl(1+(m-1)e^{-\epsilon}\bigr)^n}.
\end{equation*}

Similar to $\{P_l,l=0,\dots,n\}$, $\{Q_l^*,l=0,\dots,n-1\}$ satisfies
\begin{equation}\label{ineqQIne}
Q_l^*\le \frac{N_l}{N_{n-1}}e^{(n-1-l)\epsilon}Q_{n-1}^*.
\end{equation}
Since $\sum_{l=0}^{n-1}Q_l^*=1-\frac{N_ne^{-n\epsilon}}{\bigl(1+(m-1)e^{-\epsilon}\bigr)^n}$, taking summation over $l$ in \eqref{ineqQIne} yields
\begin{equation*}
Q_{n-1}^*\ge\frac{N_{n-1}e^{-(n-1)\epsilon}}{\bigl(1+(m-1)e^{-\epsilon}\bigr)^n}.
\end{equation*}
Using similar arguments we have
\begin{equation*}
\sum_{l=0}^{n-1}lQ_l^*\ge\sum_{l=0}^{n-2}lC_l^*+\frac{(n-1)N_{n-1}e^{-(n-1)\epsilon}}{\bigl(1+(m-1)e^{-\epsilon}\bigr)^n},
\end{equation*}
where $\{C_l^*,l=0,\dots,n-2\}$ is the optimal solution of
\begin{equation*}
\begin{aligned}
\text{min}\mspace{18mu} &&& \sum_{l=0}^{n-2}lC_l\\
\text{subject to} &&& C_l \ge 0,\mspace{74mu}l=0,1,\dots,n-2,\\
&&& \frac{C_{l-1}}{N_{l-1}}\le\frac{C_l}{N_l}e^{\epsilon},\mspace{18mu} l=1,2,\dots,n-2,\\
&&& \sum_{l=0}^{n-2}C_l=1-\frac{N_{n-1}e^{-(n-1)\epsilon}}{\bigl(1+(m-1)e^{-\epsilon}\bigr)^n}\\
&&&\mspace{86mu}-\frac{N_{n}e^{-n\epsilon}}{\bigl(1+(m-1)e^{-\epsilon}\bigr)^n}.
\end{aligned}
\end{equation*}

Continue this procedure we obtain
\begin{equation*}
\sum_{l=0}^{n}lP_l\ge \sum_{l=0}^{n}\frac{lN_le^{-(n-l)\epsilon}}{\bigl(1+(m-1)e^{-\epsilon}\bigr)^n}=\frac{n}{1+\frac{e^{\epsilon}}{m-1}}=h(\epsilon).
\end{equation*}
Therefore, for any feasible $\{p_{X\mid Y}(x\mid y),x,y\in\mathcal{D}^n\}$,
\begin{equation*}
\sum_{x\in\mathcal{D}^n}p_{X\mid Y}(x\mid y)d(x,y)=\sum_{l=0}^{n}lP_l\ge h(\epsilon),
\end{equation*}
which completes the proof of the claim.

By this claim, any feasible solution satisfies
\begin{equation*}
\sum_{x\in\mathcal{D}^n}\sum_{y\in\mathcal{D}^n}p_Y(y)p_{X\mid Y}(x\mid y)d(x,y)\ge h(\epsilon).
\end{equation*}
Therefore
\begin{equation}\label{eqDlowerbound}
D^*_\mathrm{relaxed}(\epsilon) \ge h(\epsilon).
\end{equation}
Next we prove the following claim, which gives an upper bound on the minimum distortion $D^*_\mathrm{relaxed}(\epsilon)$.
\begin{claim*}
Consider
\begin{equation}
p_{X\mid Y}(x\mid y)=\frac{e^{-\epsilon d(x,y)}}{\bigl(1+(m-1)e^{-\epsilon}\bigr)^n},\quad x,y\in\mathcal{D}^n,
\end{equation}
and any $\{p_Y(y),y\in\mathcal{D}^n\}$ with
\begin{equation*}
\sum_{y\in\mathcal{D}^n}p_Y(y)=1,\qquad p_Y(y)\ge 0,\quad\forall y\in\mathcal{D}^n.
\end{equation*}
Then $\{p_{X\mid Y}(x\mid y),x,y\in\mathcal{D}^n\}$ and $\{p_Y(y),y\in\mathcal{D}^n\}$ form a feasible solution of R-PD, and
\begin{equation}
\sum_{x\in\mathcal{D}^n}\sum_{y\in\mathcal{D}^n}p_Y(y)p_{X\mid Y}(x\mid y)d(x,y) = h(\epsilon).
\end{equation}
\end{claim*}

\noindent\textbf{Proof of the Claim.} Obviously the considered $\{p_{X\mid Y}(x\mid y),x,y\in\mathcal{D}^n\}$ and $\{p_Y(y),y\in\mathcal{D}^n\}$ satisfy constraints \eqref{conRelaxedNonneg}--\eqref{conRelaxedNonnegPY}. Therefore to prove the feasibility, we are left with constraint \eqref{conRelaxedPrivacy} and \eqref{conRelaxedNorm}. We first verify constraint \eqref{conRelaxedPrivacy}. Consider any pair of neighboring elements $x,x'\in\mathcal{D}^n$ and any $y\in\mathcal{D}^n$. Then by the triangle inequality,
\begin{equation*}
d(x,y)\le d(x',y)-d(x',x)=d(x',y)-1.
\end{equation*}
Therefore,
\begin{align*}
p_{X\mid Y}(x\mid y) &= \frac{e^{-\epsilon d(x,y)}}{\bigl(1+(m-1)e^{-\epsilon}\bigr)^n}\\
&\le \frac{e^{-\epsilon(d(x',y)-1)}}{\bigl(1+(m-1)e^{-\epsilon}\bigr)^n}\\
&=e^{\epsilon}p_{X\mid Y}(x'\mid y).
\end{align*}
Next we verify constraint \eqref{conRelaxedNorm}. For any $y\in\mathcal{D}^n$ and any integer $l$ with $0\le l\le n$, let $\mathcal{N}_l(x)$ be the set of elements with distance $l$ to $y$ as defined in \eqref{eqNl}. Then it is easy to see that $N_l\triangleq|\mathcal{N}_l(y)|=\binom{n}{l}(m-1)^l$, and for any $y\in\mathcal{D}^n$,
\begin{equation*}
\mathcal{D}^n=\bigcup_{l=0}^n\mathcal{N}_l(y).
\end{equation*}
Therefore, for any $y\in\mathcal{D}^n$,
\begin{align*}
&\mspace{21mu}\sum_{x\in\mathcal{D}^n} p_{X\mid Y}(x\mid y)\nonumber\\
&=\sum_{x\in\mathcal{D}^n}\frac{e^{-\epsilon d(x,y)}}{\bigl(1+(m-1)e^{-\epsilon}\bigr)^n}\\
&=\frac{1}{\bigl(1+(m-1)e^{-\epsilon}\bigr)^n}\sum_{l=0}^{n}\sum_{x\in\mathcal{N}_l(y)}e^{-\epsilon d(x,y)}\\
&=\frac{1}{\bigl(1+(m-1)e^{-\epsilon}\bigr)^n}\sum_{l=0}^{n}\binom{n}{l}(m-1)^l e^{-\epsilon l}\\
&=1.
\end{align*}

With feasibility verified, we can proceed to calculate the distortion. Let $g_{\epsilon}=1+(m-1)e^{-\epsilon}$. Then
\begin{align*}
&\mspace{21mu}\sum_{x\in\mathcal{D}^n}\sum_{y\in\mathcal{D}^n}p_Y(y)p_{X\mid Y}(x\mid y)d(x,y)\\
&=\frac{1}{(g_{\epsilon})^n}\sum_{y\in\mathcal{D}^n}p_Y(y)\sum_{l=0}^{n}\sum_{x\in\mathcal{N}_l(y)}e^{-\epsilon d(x,y)}d(x,y)\\
&=\frac{1}{(g_{\epsilon})^n}\sum_{y\in\mathcal{D}^n}p_Y(y)\sum_{l=0}^{n}\binom{n}{l}(m-1)^le^{-\epsilon l}l\\
&=\frac{n(m-1)e^{-\epsilon}\bigl(1+(m-1)e^{-\epsilon}\bigr)^{n-1}}{(g_{\epsilon})^n}\sum_{y\in\mathcal{D}^n}p_Y(y)\\
&=\frac{n}{1+\frac{e^{\epsilon}}{m-1}}\\
&=h(\epsilon),
\end{align*}
which completes the proof of the claim.

By this claim, there exists a feasible solution such that
\begin{equation*}
\sum_{x\in\mathcal{D}^n}\sum_{y\in\mathcal{D}^n}p_Y(y)p_{X\mid Y}(x\mid y)d(x,y) = h(\epsilon),
\end{equation*}
which implies
\begin{equation*}
D^*_\mathrm{relaxed}(\epsilon) \le h(\epsilon).
\end{equation*}
Combining this upper bound with the lower bound \eqref{eqDlowerbound} gives
\begin{equation*}
D^*_\mathrm{relaxed}(\epsilon) = h(\epsilon).
\end{equation*}
\end{proof}

\begin{lemma}\label{lem:lowerBounds}
The optimal value $D_{\mathrm{relaxed}}^*(\epsilon)=h(\epsilon)$ of R-PD implies the following lower bounds for any $D$ with $0\le D\le n$:
\begin{gather}
\epsilon_{\mathrm{i}}^*(D)\ge h^{-1}(D),\label{eq:lowerID}\\
\epsilon_{\mathrm{d}}^*(D)\ge\max\{h^{-1}(D)-\epsilon_X,0\}.\label{eq:lowerDP}
\end{gather}
\end{lemma}
\begin{proof}
First we derive the lower bound on $\epsilon_{\mathrm{i}}^*(D)$. Let $\delta$ be an arbitrary positive number. For any $D$ with $0\le D\le n$, let $\epsilon_{D,\delta}=\epsilon_{\mathrm{i}}^*(D)+\delta$. Then by the definition of $\epsilon_{\mathrm{i}}^*$, we have that $(\epsilon_{D,\delta},D)$ is achievable under identifiability. Therefore
\begin{equation*}
D\ge D_{\mathrm{i}}^*(\epsilon_{D,\delta})\ge D_{\mathrm{relaxed}}^*(\epsilon_{D,\delta})=h(\epsilon_{D,\delta}),
\end{equation*}
where $D_{\mathrm{i}}^*(\cdot)$ is the optimal value of PD-I. Since $h$ is a decreasing function, this implies that $\epsilon_{D,\delta}\ge h^{-1}(D)$. Therefore
\begin{equation*}
\epsilon_{\mathrm{i}}^*(D)\ge h^{-1}(D)-\delta.
\end{equation*}
Letting $\delta\to 0$ yields
\begin{equation*}
\epsilon_{\mathrm{i}}^*(D)\ge h^{-1}(D).
\end{equation*}

Next we derive the lower bound on $\epsilon_{\mathrm{d}}^*(D)$ using arguments similar to those in the proof of the lower bound on $\epsilon_{\mathrm{i}}^*(D)$. Let $\delta$ be an arbitrary positive number. For any $D$ with $0\le D\le n$, let $\epsilon_{D,\delta}=\epsilon_{\mathrm{d}}^*(D)+\delta$. Then by the definition of $\epsilon_{\mathrm{d}}^*$, we have that $(\epsilon_{D,\delta},D)$ is achievable under differential privacy. Therefore
\begin{equation*}
D\ge D_{\mathrm{d}}^*(\epsilon_{D,\delta})\ge D_{\mathrm{relaxed}}^*(\epsilon_{D,\delta}+\epsilon_X)=h(\epsilon_{D,\delta}+\epsilon_X),
\end{equation*}
where $D_{\mathrm{d}}^*(\cdot)$ is the optimal value of PD-DP. Since $h$ is a decreasing function, this implies that $\epsilon_{D,\delta}+\epsilon_X\ge h^{-1}(D)$. Therefore
\begin{equation*}
\epsilon_{\mathrm{d}}^*(D)\ge h^{-1}(D)-\epsilon_X-\delta.
\end{equation*}
Letting $\delta\to 0$ yields
\begin{equation*}
\epsilon_{\mathrm{d}}^*(D)\ge h^{-1}(D)-\epsilon_X.
\end{equation*}
Since the privacy level is nonnegative, we obtain the lower bound in \eqref{eq:lowerDP}.
\end{proof}

\begin{lemma}\label{lem:lower2ID}
The privacy--distortion function $\epsilon_{\mathrm{i}}^*$ of a database $X$ is bounded from below as
\begin{equation}
\epsilon_{\mathrm{i}}^*(D)\ge\epsilon_X
\end{equation}
for any $D$ with $0\le D\le n$, where $\epsilon_X$ is the constant defined in \eqref{eqEpsX}.
\end{lemma}
\begin{proof}
Suppose by contradiction that there exists a $D$ with $0\le D\le n$ such that $\epsilon_{\mathrm{i}}^*(D)<\epsilon_X$. Let $\delta$ be an arbitrary positive number with $0<\delta<\epsilon_X-\epsilon_{\mathrm{i}}^*(D)$, and let $\epsilon=\epsilon_{\mathrm{i}}^*(D)+\delta$. Then $\epsilon<\epsilon_X$ and $(\epsilon,D)$ is achievable under identifiability. Consider the mechanism that achieves $(\epsilon,D)$. Then by the requirement of identifiability, for any neighboring $x,x'\in\mathcal{D}^n$ and any $y\in\mathcal{D}^n$,
\begin{equation}\label{eqPostPrivacy}
p_{X\mid Y}(x\mid y)\le e^{\epsilon}p_{X\mid Y}(x'\mid y).
\end{equation}
Let $p_Y(\cdot)$ be the pmf of the output $Y$. Then $p_Y(y)\ge 0$ for any $y\in\mathcal{D}^n$. Therefore, multiplying both sides of \eqref{eqPostPrivacy} by $p_Y(y)$ and taking summation over $y\in\mathcal{D}^n$ yield
\begin{equation*}
\sum_{y\in\mathcal{D}^n}p_{X\mid Y}(x\mid y)p_Y(y)\le \sum_{y\in\mathcal{D}^n}e^{\epsilon}p_{X\mid Y}(x'\mid y)p_Y(y),
\end{equation*}
which implies
\begin{equation*}\label{eqContradict}
p_X(x)\le e^{\epsilon} p_X(x').
\end{equation*}
Then there do not exist neighboring $x,x'\in\mathcal{D}^n$ with $p_X(x)=e^{\epsilon_X}p_X(x')$ since $\epsilon<\epsilon_X$, which contradicts with the definition of $\epsilon_X$ in \eqref{eqEpsX}.
\end{proof}

\begin{lemma}\label{lem:DP-I}
For $\epsilon\ge\widetilde{\epsilon}_X$, the mechanism $\mathcal{E}_i$ defined in \eqref{eqPostExpMec} satisfies $\epsilon$-identifiability, and the distortion of $\mathcal{E}_i$ is given by $\mathbb{E}[d(X,Y)]=h(\epsilon)$.
\end{lemma}
\begin{proof}
Consider any $\epsilon\ge\widetilde{\epsilon}_X$. Then under the mechanism $\mathcal{E}_i$, the posterior probability for any $x,y\in\mathcal{D}^n$ is given by
\begin{equation*}
p_{X\mid Y}(x\mid y)=\frac{p_{Y\mid X}(y\mid x)p_X(x)}{p_Y(y)}=\frac{e^{-\epsilon d(x,y)}}{\bigl(1+(m-1)e^{-\epsilon}\bigr)^n}.
\end{equation*}
As shown in the proof of Lemma~\ref{lem:R-PD}, this $\{p_{X\mid Y}(x\mid y),x,y\in\mathcal{D}^n\}$ and the corresponding $\{p_Y(y),y\in\mathcal{D}^n\}$ form an optimal solution of the relaxed optimization problem R-PD. Following the same arguments as in the proof of Lemma~\ref{lem:R-PD} we can conclude that $\mathcal{E}_i$ satisfies $\epsilon$-identifiability, and the distortion of $\mathcal{E}_i$ is given by $\mathbb{E}[d(X,Y)]=h(\epsilon)$.
\end{proof}

\begin{lemma}\label{lemMecE}
The mechanism $\mathcal{E}_d$ defined in \eqref{eqExpMec} satisfies $\epsilon$-differential privacy, and the distortion of $\mathcal{E}_d$ is given by $\mathbb{E}[d(X,Y)]=h(\epsilon)$.
\end{lemma}
\begin{proof}
Under mechanism $\mathcal{E}_d$, $\{p_{Y\mid X}(y\mid x),x,y\in\mathcal{D}^n\}$ has the same form as the posteriors under mechanism $\mathcal{E}_i$. Therefore still by similar arguments as in the proof of Lemma~\ref{lem:R-PD}, $\mathcal{E}_d$ satisfies $\epsilon$-differential privacy, and the distortion of $\mathcal{E}_d$ is given by $\mathbb{E}[d(X,Y)]=h(\epsilon)$.
\end{proof}

\noindent\textbf{Derivation of the Stationarity Condition}

The Lagrangian of the optimization problem PD-MIP can be written as
\begin{align*}
&\mspace{24mu}J(\{p_{Y\mid X}(y\mid x),x,y\in\mathcal{D}^n\})\nonumber\\
&=\sum_{x\in\mathcal{D}^n}\sum_{y\in\mathcal{D}^n}p_X(x)p_{Y\mid X}(y\mid x)\ln\frac{p_{Y\mid X}(y\mid x)}{\mspace{-6mu}\displaystyle\sum_{x'\in\mathcal{D}^n}p_X(x')p_{Y\mid X}(y\mid x')}\\
&\mspace{24mu}+\lambda\Biggl(\sum_{x\in\mathcal{D}^n}\sum_{y\in\mathcal{D}^n}p_X(x)p_{Y\mid X}(y\mid x)d(x,y)-D\Biggr)\\
&\mspace{24mu}+\sum_{x\in\mathcal{D}^n}\mu(x)\Biggl(\sum_{y\in\mathcal{D}^n}p_{Y\mid X}(y\mid x)-1\Biggr)\\
&\mspace{24mu}-\sum_{x\in\mathcal{D}^n}\sum_{y\in\mathcal{D}^n}\eta(x,y)\;p_{Y\mid X}(y\mid x),
\end{align*}
where $\lambda,\mu(x),\eta(x,y),\;x,y\in\mathcal{D}^n$ are the Lagrangian multipliers. Consider the KKT conditions of this optimization problem. The stationarity condition is given by
\begin{equation*}
\frac{\partial J}{\partial p_{Y\mid X}(y\mid x)}=0,\quad x,y\in\mathcal{D}^n,
\end{equation*}
which is equivalent to
\begin{equation*}\label{eqStation}
p_X(x)\ln\frac{p_{Y\mid X}(y\mid x)}{p_Y(y)}+\lambda p_X(x)d(x,y)\\
+\mu(x)-\eta(x,y)=0,\quad x,y\in\mathcal{D}^n.
\end{equation*}

\bibliographystyle{IEEEtran}
\bibliography{differential-privacy}{}
\end{document}